\documentclass[11pt]{article}
\usepackage{amsmath, amssymb, enumerate}
\usepackage{amsfonts, setspace}
\usepackage{cite}
\usepackage{graphicx, subeqnarray, flushend, amsthm}

\usepackage[margin=1in]{geometry}

\newtheorem{lemma}{Lemma}
\newtheorem{theorem}{Theorem}
\newtheorem{condition}{Condition}

\theoremstyle{definition}
\newtheorem{definition}{Definition}
\newtheorem{remark}{Remark}

\newcommand{\linf}{$\mathcal{L}_\infty$}
\newcommand{\ltwo}{$\mathcal{L}_2$}
\newcommand{\iqc}{$\rm{\delta}$QC}
\newcommand{\imm}{$\rm{\delta}$MM}

\newcommand\ignore[1]{{}}
\title{\ltwo~Observers for Nonlinear Systems with Unbounded Unknown Inputs}
\author{Martin Corless$^1$ and Ankush Chakrabarty$^{2}$
	\thanks{$^1$School of Aeronautics and Astronautics, Purdue University, West Lafayette, IN, USA. Email: \texttt{corless@purdue.edu}}%
	\thanks{$^2$Control and Dynamical (CD) Systems Group, Mitsubishi Electric Research Laboratories, Cambridge, MA, USA. Corresponding author: A.~Chakrabarty. Phone: +1~(617)~758-6175. Email: \texttt{chakrabarty@merl.com}.}%
}

\usepackage[dvipsnames]{xcolor}
\usepackage{hyperref}
\usepackage{cleveref}

\newcommand\myshade{85}
\colorlet{mylinkcolor}{RoyalBlue}
\colorlet{mycitecolor}{SeaGreen}
\colorlet{myurlcolor}{Aquamarine}

\hypersetup{
	linkcolor  = mylinkcolor!\myshade!black,
	citecolor  = mycitecolor!\myshade!black,
	urlcolor   = myurlcolor!\myshade!black,
	colorlinks = true,
}
\title{\bf \ltwo~Observers for a Class of Nonlinear Systems with Unknown Inputs \\ \bigskip
\normalsize [Full version of paper presented in European Control Conference, 2019]}
\date{November 16, 2018}
\begin{document}
	
	\maketitle
	\begin{abstract}                
	We consider the problem of estimating the state and unknown input for a large class of nonlinear systems subject to  unknown exogenous inputs. The  exogenous inputs themselves are modeled as being generated by a nonlinear system subject to unknown inputs.
	The nonlinearities considered in this work are characterized by multiplier matrices that include many commonly encountered nonlinearities.
	We obtain a linear matrix inequality (LMI), that, if feasible, provides the gains for an observer which results in certified $\mathcal{L}_2$ performance of the error dynamics associated with the observer.
	We also present conditions which guarantee that the $\mathcal{L}_2$ norm of the error can be made arbitrarily small and investigate conditions for feasibility of the proposed LMIs.
	\end{abstract}
	
	
	\section{Introduction}
	\label{sec:intro}
	Exogenous unknown inputs acting on a  dynamical system (plant)  can result in compromised safety and degraded performance.
	One way to protect a system against such unknown attacks is by employing unknown input observers
	(UIOs), as reported in \cite{Teixeira2014Thesis} and \cite{chakrabarty2017delayed}. 
	Common  estimation frameworks for systems in which one assumes stochastic models for  the unknown exogeneous  input include Kalman filtering~\cite{keller1999two} and minimum variance filters~\cite{khemiri2011novel}. For  unknown exogeneous  inputs where underlying statistics are not available and cannot be guessed, methods that have proven effective include: adaptive estimation~\cite{zhang2010fault}, sliding mode observers~\cite{floquet2007sliding,fridman2008higher}, and observers that minimize the system's input-output gain such as $\mathcal{H}_\infty$ observers~\cite{Chakrabarty2016state,chen2007observer,zemouche2008observers,zemouche2016HinftyObserver}. Recent work has produced many effective methods for generating unknown input observers for nonlinear systems; see for example~\cite{Bejarano2012,Chen2006,Ha2004,Witczak2016,barbot2007algebraic,zemouche2009nonlinear,Chakrabarty2016sufficient,chakrabarty2018state}.
	
	A common underlying assumption in many of the cited works is that the  unknown exogeneous  input is bounded. One  way of relaxing that assumption is by using an extended state observer, that is, by appending the exogenous input  to the system state.
	Exogenous input estimation via an extended state observer has been successful in various practical systems, including robotic systems~\cite{su2004task}, electric drive systems~\cite{liu2012speed}, power electronics~\cite{wang2015extended}, and avionics~\cite{xia2011attitude}. These exogenous inputs could be {completely unknown}, or {partially unknown}. In this paper, we refer to partially unknown inputs as exogenous inputs that have been generated by a completely unknown input acting on 
	in~\cite{francis1976internal,johnson1971accomodation}. 
	Although prior investigations into extended state observer design for  estimation  with unknown exogenous inputs has yielded useful results~\cite{li2012generalized,godbole2013performance,wei2014composite}, the authors assume that the inputs  are bounded, and tackle linear systems or linearized versions of nonlinear systems; the sparsity of results on observers for  nonlinear  systems with unknown deterministic exogenous  inputs   motivates the present paper. 
	
	In this paper, we propose  extended state observers to   estimate  the state and the unknown exogenous input for   nonlinear systems 
	whose nonlinearities satisfy so-called incremental quadratic constraints \cite{Acikmese2011,Chakrabarty2016}. 
	Such nonlinearities encompass a wide range of nonlinearities including globally Lipschitz, one-sided Lipschitz, monotonic and other commonly occurring nonlinearities. Also, the exogenous input can be unbounded. Observer design is based on a linear matrix inequality which we demonstrate is satisfied by a large class of commonly encountered nonlinear systems.
	The observers guarantee that the input-output system from exogenous input to observer error is $\mathcal{L}_2$-stable with a specific gain; for linear systems this gain is an upper bound on the $\mathcal{H}_\infty$ norm of the system.
	We also present conditions which guarantee that  an arbitrarily small \ltwo~gain can be achieved.
	

	\section{Problem statement}\label{sec:ps}
	\subsection{Systems under consideration}
	Consider a nonlinear time-varying  system (the plant)  described by
	\begin{subequations}
		\label{eq:plant2}
		\begin{alignat}{8}
		\dot x \;&=\; A x \;&+&\;B_f f_1(t, y, q_1) \;&+&\; g_x(t,y) \;&+&\; B w,			\\
		q_1 \;&=\; C_{q_1} x \;&+&\; D_{q_1f} f_1(t,y,q_1) \;&+&\; g_{q_1}(t,y) \;&+&\; D_{q_1} w,		\\
		y \;&=\; C x \;&+&\; D_f f_1(t,y,q_1) \;&+&\; g_y(t,y) \;&+&\; D w.
		\end{alignat}
	\end{subequations}
	Here, $t \in\mathbb{R}$ is the {time variable}, $x(t)\in\mathbb{R}^{n_x}$ is  the {state}, 
	$y(t)\in\mathbb{R}^{n_y}$ is the {measured output} and  
	$w(t)\in\mathbb{R}^{n_w}$ models the disturbance input and  the measurement noise combined into one term; we refer to it as the
	exogenous input; this is unknown at every $t$.
	The vector $f_1(t,y, q_1) \in\mathbb{R}^{n_{f_1}}$ models  nonlinearities  of known structure, but because this term depends on the state $x$ (through $q_1$), it cannot be instantaneously determined from measurements. The vector $q_1\in\mathbb{R}^{n_{q_1}}$ is a state-dependent argument of the nonlinearity $f_1$.
	The vectors $g_x(t,y) \in\mathbb{R}^{n_x}$,
	$g_q(t,y) \in\mathbb{R}^{n_q}$ and $g_y(t,y) \in\mathbb{R}^{n_y}$ represent nonlinearities which can be calculated instantaneously from measurements.
	An example is $g_x(t,y) = u(t)$ where $u(t)$ is a control input.
	All the matrices are  constant and of appropriate dimensions.
	
	We consider the general case in which   the exogenous input $w$ is generated by the following 
	{ nonlinear exogenous input model:}
	\begin{subequations}
		\label{eq:internal_model_w}
		\begin{alignat}{8}
		\dot x_m &= A_m x_m \;&+&\; B_{mf} f_2(t, q_2) \;&+&\; B_m v				\\
		q_2 &= C_{q_2} x_m \;&+&\; D_{q_2f} f_2(t, q_2) \;&+&\; D_{q_2} v			\\
		w &= C_m x_m \;&+&\; D_{wf} f_2(t, q_2) \;&+&\; D_m v,
		\end{alignat}
	\end{subequations}
	where $x_m(t)\in\mathbb{R}^{n_m}$ is a exogenous input model state, 
	$v(t)\in\mathbb{R}^{n_v}$ is another  unknown exogenous input signal and $f_2$ is a known nonlinearity.

	\begin{definition}[$\mathcal L_2$ signal]
		We say  that a signal $s(\cdot) : [t_0, \infty) \rightarrow \mathbb{R}^p$ is $\mathcal{L}_2$ if 
		$
		\int_{t_0}^\infty  \|s(t)\|^2 \,\mathrm{d} t
		$
		is finite 
		where $||s(t)\|$ is the usual Euclidean norm of $s(t)$ and we define its $\mathcal{L}_2$ norm by
		\begin{align}
		\|s(\cdot)\|_2 =  \left(\int_{t_0}^\infty  \|s(t)\|^2 \,\mathrm{d} t\right)^{\frac{1}{2}}.
		\end{align}
		When we say that a signal is bounded, we mean that it is an $\mathcal{L}_2$ signal.
	\end{definition}
	
	\begin{remark}
		\label{rem:boundedDer}
		The model~\eqref{eq:internal_model_w}  is used to reflect partial knowledge regarding the unknown input $w$.
		For example, 
		if $w$ is an  unknown input with an unknown  derivative  which is $\mathcal{L}_2$ 
		it  can be described with  model~\eqref{eq:internal_model_w} with a bounded $v$; specifically,
		\begin{align}
		\label{eq:simpleW}\dot{x}_m= v ,			\qquad w = x_m
		\end{align}
		where  $v= \dot{w}$ is $\mathcal{L}_2$.
		An example is $$w(t) =a + \ln(1+bt)$$ where $a$ and $b$ are unknown constants.
	\end{remark}

	In this paper, we characterize  nonlinearities  via their incremental multiplier matrices.
	\begin{definition}
		[Incremental Multiplier Matrices]
		A symmetric matrix $M \in \mathbb{R}^{(n_q+n_f)\times (n_q+n_f)}$  is an {\sf incremental multiplier matrix}~(\imm) for 
		$f$ if it satisfies the following {\sf incremental quadratic constraint}~(\iqc) for all 
		$ t \in\mathbb{R}$,   $y \in \mathbb{R}^{n_y}$ and $q_1, q_2 \in\mathbb{R}^{n_q}$:
		\begin{equation}
		\label{eq:iqc}
		\begin{bmatrix}
		\Delta q \\ \Delta f
		\end{bmatrix}^\top M \begin{bmatrix}
		\Delta q \\ \Delta f
		\end{bmatrix} \ge 0,
		\end{equation}
		where $\Delta q\triangleq q_1 - q_2$ and $\Delta f \triangleq f(t,y,q_1) - f(t,y, q_2)$.
	\end{definition}
	The utility of characterizing nonlinearities using incremental multipliers is that our observer design strategy applies to a broad class of nonlinear systems. \imm~for many common nonlinearities are provided in~\cite{Acikmese2011,Chakrabarty2016}.
	
	\subsection{Problem statement}
	Ideally, we wish to obtain observers that provide an estimate  of 
	$x$ and $w$.
	To this end we define the augmented state
	\begin{equation}\label{eq:aug_state}
	\xi \triangleq \begin{bmatrix}
	x \\ x_m
	\end{bmatrix}
	\end{equation}
	and look for observers to obtain an estimate $\hat{\xi}$ of $\xi$.
	With
	\[
	\begin{bmatrix}
	\hat x \\ \hat x_m
	\end{bmatrix}
	= \hat{\xi}
	\]
	$\hat{x}$ and $\hat{x}_m$ will be the observer estimates of $x$ and $x_m$, respectively.
	An estimate of $w$ can be achieved if $D_{wf}=D_m = 0$. In this case, an estimate of the unknown input $w$ is given by
	\[
	\hat{w} = C_m \hat{x}
	\]
	This occurs in the special case when $w$ has a bounded derivative; see Remark \ref{rem:boundedDer}.
	
	Let
	\begin{equation}
	\label{eq:error}
	e = \hat \xi - \xi
	\end{equation}
	denote the { estimation error} and suppose
	that 
	\begin{equation}\label{eq:perf_out}
	z = H e
	\end{equation}
	is  a user-defined  { performance output} 
	associated with the observer where $z\in \mathbb{R}^{n_z}$.
	%
	As we  demonstrate below,
	a proposed observer generates  an { error system} that can be described by
	\begin{subequations}
		\label{eq:gensys}
		\begin{align}
		\dot{e} &= F(t,e,v),	\label{eq:gensysa}\\
		z&= G(t,e,v).  	\label{eq:gensysb}		
		\end{align}	
	\end{subequations}
	We want this system to have the following performance 		
	with performance level $\gamma$.
	
	\begin{definition}
		\label{defn1} 
		Let $\gamma$ be a non-negative  real scalar. The input-output system~\eqref{eq:gensys} 
		is {\sf globally uniformly \ltwo-stable with performance level $\gamma$} if it has the following properties.
		\begin{enumerate}[{(P1)}]
			\item
			{\sf  Global uniform exponential stability with zero input.}
			The zero-input system ($v \equiv0$) is globally uniformly exponentially stable about the origin.
			%
			\item
			{\sf  Global uniform boundedness of the error state.}
			\label{defn1:gub}
			For  every initial condition $e(t_0) =e_0 $,
			and every  \ltwo~unknown input $v(\cdot)$, there exists 
			$\beta_1(e_0,\|v(\cdot)\|_2)$ such that
			$$
			\|e(t)\| \le \beta_1(e_0,\| v (\cdot)\|_2)$$ for all $t \ge t_0$.
			\item
			\label{defn1:zero}
			{ \sf Output response.} 
			For  every initial condition $e(t_0) =e_0 $,
			and every \ltwo~unknown input $v(\cdot)$, there exists 
			$\beta_2(e_0,\|v(\cdot)\|_2)$ such that
			$$\|z(\cdot)\|_2  \le \beta_2(e_0,\| v (\cdot)\|_2)$$
			and 
			$$\beta_2(0, \| v (\cdot)\|_2) \le  \gamma \|v(\cdot)\|_2.$$
		\end{enumerate}
	\end{definition}

	\section{Proposed observers}
	With the augmented state $\xi$ given by \eqref{eq:aug_state},
	we obtain  the {\sf  augmented plant}:
	\begin{subequations}
		\label{eq:aug_sys}
		\begin{alignat}{10}
		\dot \xi &= \mathcal{A}\xi \;&+&\; \mathcal{B}_f f \;&+&\; {\tilde g}_{\xi} \;&+&\;\mathcal{B} v\\
		\label{eq:q}
		q&= \mathcal{C}_q \xi\;&+&\; \mathcal{D}_{qf} f \;&+&\; \tilde{g}_q \;&+&\; \mathcal{D}_{q} v\\
		y&= \mathcal{C } \xi	\;&+&\; \mathcal{D}_f f \;&+&\;\tilde{g}_y \;&+&\; \mathcal{D} v,
		\end{alignat}
	\end{subequations}
	where
	\[
	f(t, y, q) \triangleq 
	\begin{bmatrix}
	f_1(t,y,q_1) \\ f_2(t,q_2)
	\end{bmatrix}
	\,\qquad \tilde{g}_*(t,y)  = \begin{bmatrix}
	g_*(t,y) \\ 0
	\end{bmatrix}
	\]
	with $ q =
	\begin{bmatrix}
	q_1^\top& q_2^\top
	\end{bmatrix}^\top$
	and
	\begin{subequations}
		\label{eq:augGen}
		\begin{alignat}{2}
		{\mathcal A} &= \begin{bmatrix}
		A & B C_m \\ 0 & A_m
		\end{bmatrix},\,
		{\mathcal B}_f =\begin{bmatrix}
		B_f & BD_{wf} \\ 0 & B_{mf}
		\end{bmatrix},\\
		{\mathcal B} &=\begin{bmatrix}
		B D_m \\ B_m 
		\end{bmatrix},\,
		\mathcal C_{q}= \begin{bmatrix}
		C_{q_1} & D_{q_1} C_m \\ 0 & C_{q_2}
		\end{bmatrix}, \\		
		\mathcal D_{qf} &=\begin{bmatrix} D_{q_1f} & D_{q_1}D_{wf} \\ 0 & D_{q_2f} \end{bmatrix}, \, \mathcal{D}=DD_m		,\\
		{\mathcal D}_{q} &=\begin{bmatrix} D_{q_1} D_m \\ D_{q_2} \end{bmatrix},\,
		{\mathcal C} = \begin{bmatrix}
		C & DC_m
		\end{bmatrix},	\\							
		\mathcal{D}_f &= \begin{bmatrix}D_f & DD_{wf}\end{bmatrix}.
		\end{alignat}
	\end{subequations}
	
	In view of the above augmented plant, we propose the following { observer}:
	\begin{subequations}
	\label{eq:obs}
	\begin{alignat}{7}
		\dot{\hat \xi} &= \mathcal A \hat \xi \;&+&\; \mathcal B_f f (t, y, \hat q)	\;&+&\;  \tilde{g}_{\xi}(t, y) \;&+&\; L_1( \hat{y}-y)			\\
		\hat q &= \mathcal C_q \hat{\xi} \;&+&\; \mathcal{D}_{qf} f(t,y,\hat q) \;&+&\;  \tilde{g}_q(t,y) \;&+&\; L_2 ( \hat{y}-y)		\\
		\hat{y} &= \mathcal{C }\hat{\xi} \;&+&\; \mathcal{D}_f f(t,y, \hat{q})\;&+&\; \tilde{g}_y(t,y),\;&{}&\;
	\end{alignat}
	\end{subequations}
	where $\hat \xi$ is an estimate of the augmented state $\xi$.
	Basically,  the proposed observer is a copy  of the augmented plant along with two correction terms $L_1( \hat{y}-y)$ and $L_2 ( \hat{y}-y)$.
	
	{\sf Observer gains}  $L_1$, $L_2$ that yield the desired performance  can be obtained using the  following result.
	
	\begin{theorem}
		\label{thm:1}
		Consider the augmented plant~\eqref{eq:aug_sys}
		along with performance output given by~\eqref{eq:perf_out}.
		Suppose  that there exist matrices $\mathcal P=\mathcal P^\top \succ 0$, $\mathcal Y$,  $L_2$, 
		an incremental multiplier matrix $M$ for    $f$,  and    scalars $\alpha, \mu_1 >0, \mu_2\ge  0$ such that
		\begin{align}
		\label{eq:thm1a}
		\Xi + \Gamma^\top M \Gamma \preceq 0			
		\end{align}
		where 
		\[
		\Xi= \begin{bmatrix}
		\Phi_\alpha +\mu_1H^\top H
		& \mathcal P \mathcal B_f + \mathcal Y \mathcal D_f& - \mathcal P \mathcal B- \mathcal Y \mathcal D\\
		\star & 0 & 0 \\
		\star &\star &-\mu_2I
		\end{bmatrix},
		\]
		with
		\begin{equation}
		\label{eq:PhiAlpha}
		\Phi_\alpha = \mathcal P \mathcal A + \mathcal A^\top \mathcal P + \mathcal Y \mathcal C + \mathcal C^\top \mathcal Y^\top\!+\!2\alpha \mathcal P,
		\end{equation}
		and
		\begin{equation*}
		\Gamma = \begin{bmatrix}
		\mathcal C_q + L_2 \mathcal C &  \mathcal{D}_{qf} & -\mathcal{D}_q - L_2 \mathcal{D}  \\
		0 &  I & 0 
		\end{bmatrix}.
		\end{equation*}
		Consider now observer~\eqref{eq:obs} with gains 
		\begin{equation}
		L_1 =  \mathcal P^{-1} \mathcal Y
		\end{equation}
		and $L_2$.
		Then, for any initial condition $e(t_0) = e_0$  with $t_0\in \mathbb{R}$ and $e_0 \in \mathbb{R}^{n_e}$ and any $\mathcal{L}_2$ exogenous input $v(\cdot):[t_0, \infty) \rightarrow \mathbb{R}^{n_v}$,
		\begin{align}
		\label{eq:L_2}
		\|z(\cdot)\|_2 \le \sqrt{\beta_2/\mu_1} \|e_0\| + \sqrt{\mu_2/\mu_1} \|v(\cdot)\|_2
		\end{align}
		and
		\begin{align}
		\label{eq:eBound}
		\|e(t)\|   \le e^{- \alpha (t-t_0)}\sqrt{\beta_1/\beta_2} \|e_0\| +\sqrt{\beta_1/\mu_2} \|v(\cdot)\|_2		
		\end{align}
		for all $t\ge t_0$.
		Hence the  error dynamics with   performance output $z = He$  are \ltwo-stable with performance level 
		\begin{align}
		\gamma =\sqrt{\mu_2/\mu_1}
		\end{align}
	\end{theorem}
	A proof is given in Section \ref{sec:Th1Proof}.
	
	\begin{remark}
		Note that, with $\alpha$ and $L_2$ fixed,  
		the matrix inequalities in Theorem~\ref{thm:1} are linear in $\mathcal Y$, $\mathcal P$, $M$, and $\mu_1, \mu_2$. Only the \textit{structure} of $M$ has to be determined \textit{a priori} for the given nonlinearity $f$; its exact value is obtained by solving the LMI~\eqref{eq:thm1a}.
	\end{remark}
	\begin{remark}
		Although in the inequality~\eqref{eq:thm1a} we require $L_2$ to be fixed, the problem can be reposed with 
		variable $L_2$. In fact, the entirety of Section~IV in~\cite{Chakrabarty2016} is devoted to computing $L_1$ and $L_2$ simultaneously using convex programming, by exploiting the structure of the incremental multiplier matrices for the given nonlinearity.
	\end{remark}
	
	\begin{remark}\label{rk:optimal_mu}
		To get optimal estimation performance, one can  let $\mu_1 =1$ and  formulate the generalized eigenvalue problem
		\begin{equation}\label{eq:gevp}
		\mu_2^\star = \arg\min \; \mu_2 \quad \text{subject to:~\eqref{eq:thm1a}}
		\end{equation}
		to obtain a minimal $\gamma$ while line searching over $\alpha$ in some bounded set $[0, \alpha_{\max}]$.
	\end{remark}
	\begin{remark}
		Recall the class of inputs discussed in Remark \ref{rem:boundedDer}.
		Recalling  Definition~\ref{defn1}~(P3), it follows from Theorem~\ref{thm:1}  that, for zero initial state, a proposed observer results in
		\[
		\|z(\cdot)\|_2 \le \gamma \|\dot{w}(\cdot)\|_2.
		\]
		Thus, $w$ can be unbounded.
		The bound on $\|\dot w(\cdot)\|_2$ does not explicitly need to be known by the designer in order to construct the observer. However, if known, then a  bound on the performance output can be calculated.
	\end{remark}

	\subsection{Existence of observers with desired $\mathcal{L}_2$ performance}
	
	Here, we  present conditions which guarantee  the existence of observers whose error dynamics are \ltwo-stable. 
	
	\begin{lemma}\label{thm:finite_gamma_nonlin}
		Suppose that there exist matrices $\mathcal P=\mathcal P^\top\succ 0$, $\mathcal Y$,  $L_2$, 
		an incremental multiplier matrix $M$ for $f$, and a scalar $\bar{\alpha} >0$  such that
		\begin{equation}
		\label{eq:barPhiLMI}
		\begin{bmatrix}
		\Phi_{\bar \alpha}&	\mathcal P\mathcal B_f	+\mathcal{Y}\mathcal{D}_f\\
		\star		&0
		\end{bmatrix}
		+\bar \Gamma^\top  M \bar \Gamma \preceq 0,
		\end{equation}
		and
		\begin{equation*}
		\bar \Gamma = \begin{bmatrix}
		\mathcal C_q +  L_2 \mathcal C & \mathcal{D}_{qf} \\ 0 & I
		\end{bmatrix}.
		\end{equation*}
		Then, for any performance output $z=He$ and any positive  ${\alpha} <\bar{\alpha}$,  there exist  positive scalars $\mu_1, \mu_2$  such that 
		\eqref{eq:thm1a} holds.
	\end{lemma}
	
	\begin{proof}
		Suppose 
		\eqref{eq:barPhiLMI} holds.
		Choosing any positive $\alpha < \bar{\alpha}$,
		there exist positive scalars $\mu_1, \mu_2 $ such that
		$N \succ 0$ and 
		
		\begin{equation}
		\label{eq:Ineqmu1}
		\mu_1H^\top H + \Xi_{13}N^{-1}\Xi_{13}^\top  \preceq 2(\bar{\alpha} -\alpha)\mathcal{P},
		\end{equation}
		where
		\begin{equation}
		\label{eq:N}
		N=\mu_2 I -\begin{bmatrix} \mathcal{D}_q +L_2\mathcal{D}\\0\end{bmatrix}^\top M\begin{bmatrix} \mathcal{D}_q +L_2\mathcal{D} \\ 0\end{bmatrix},
		\end{equation}
		and 
		$\Xi_{13} = \mathcal{P}\mathcal{B}+\mathcal{Y}\mathcal{D}$.
		Using Schur complements,~\eqref{eq:thm1a}  is equivalent to 
		\begin{equation}
		\label{eq:PhiLMI}
		\begin{bmatrix}
		\Phi_\alpha +	\mu_1 H^\top H + \Xi_{13}N^{-1}\Xi_{13}^\top	&	\mathcal P\mathcal B_f	+\mathcal{Y}\mathcal{D}_f\\
		\star		&0
		\end{bmatrix}
		+\bar \Gamma^\top  M \bar \Gamma \preceq 0,
		\end{equation}
		It follows from \eqref{eq:Ineqmu1} that
		\[
		2\alpha \mathcal{P} +\mu_1H^\top H  +\Xi_{13}N^{-1}\Xi_{13}^\top \preceq 2 \bar {\alpha} \mathcal P.
		\]
		Thus, $$\Phi_\alpha +	\mu_1 H^\top H + \Xi_{13}N^{-1}\Xi_{13}^\top\preceq \Phi_{\bar \alpha}$$ and
		\eqref{eq:thm1a} holds.
	\end{proof}
	
	In characterizing a solution to a problem in terms of LMI's one must show that the LMI's are feasible for a significant class of systems.
	Here we show that this is the case for the LMIs presented here.
	For example, consider the case in which $\mathcal{D}_{qf} =0$ and $f$ is globally Lipschitz in the sense that
	$$\|f(t,y,\tilde{q})-f(t,y,q)\| \le \kappa \|\tilde{q}-q\|$$ for all $t,y,q,\tilde{q}$ for some $\kappa >0$.
	Here we claim that if $\kappa$ is sufficiently small then,  then there is a solution to LMI
	\eqref{eq:barPhiLMI} if 
	$(A,C)$ and $(A_m, C_m)$ are detectable and 
	the following condition is satisfied.
	\begin{condition}
		\label{cond1}
		The matrix 
		\begin{equation}
		\label{eq:condn_all_lambdam}
		\begin{bmatrix}
		A - \lambda I & B \\ C&D
		\end{bmatrix} 
		\end{equation}
		has full column rank
		for every eigenvalue $\lambda$ of $A_m$ with non-negative real part.
	\end{condition}
	Recall that 
	a pair $(C,A)$ is detectable if
	the matrix
	\[
	\mathrm{rank}\, \begin{bmatrix}
	A-\lambda I \\ C 
	\end{bmatrix} 
	\]
	has full column rank for every $\lambda \in \mathbb{C}$ with non-negative real part.
	
	To prove the above claim, we first note that an incremental  multiplier matrix for $f$ is given by
	$$
	M=
	\left[\begin{array}{cc}
	I	&0\\
	0	&-\kappa^{-2} I
	\end{array}\right]
	$$
	and, with $L_2 =0$,  \eqref{eq:barPhiLMI} reduces to
	\[
	\begin{bmatrix}
	\Phi_{\bar \alpha}
	& \mathcal P\mathcal B_f +\mathcal{Y}\mathcal{D}_f +  \mathcal{C}_q^\top \mathcal{C}_q \\
	\star & -\kappa^{-2}I
	\end{bmatrix}
	\preceq 0
	\]
	where $\Phi_{\bar \alpha}$  is given by \eqref{eq:PhiAlpha}.
	This is equivalent to
	\[
	\Phi_{\bar \alpha} +\kappa^2(\mathcal P\mathcal B_f +\mathcal{Y}\mathcal{D}_f +  \mathcal{C}_q^\top\mathcal{C}_q)(\mathcal P\mathcal B_f +\mathcal{Y}\mathcal{D}_f +  \mathcal{C}_q^\top\mathcal{C}_q)^\top \preceq 0
	\]
	If $\Phi_{\bar \alpha} \prec 0$, the above inequality is satisfied when $\kappa >0$ is sufficiently small.
	It follows from Lemmas \ref{lem:detect} and \ref{lem:detect1} (given later) that, if 
	$(A,C)$ and $(A_m, C_m)$ are detectable and Condition \ref{cond1} holds
	then, there exist matrices  $\mathcal P=\mathcal P^\top\succ 0$ and $\mathcal Y$  such that $\Phi_{\bar \alpha} \prec 0$.%
	
	\subsection{Estimating  with  arbitrarily small error}
	Here, we provide conditions which guarantee that one can estimate the plant state and exogenous input to any arbitrary accuracy, that is, 
	for  any performance output $z=He$, one can achieve  any desired level of performance $\gamma >0$.
	The result also provides 
	a method of computing observer gain matrices $L_1$ and $L_2$ to achieved the desired performance.
	
	\begin{theorem}
		\label{thm:uio_reject}
		Suppose  there exist matrices $\mathcal P=\mathcal P^\top \succ 0 $, $\mathcal Y$,  $L_2$, $\mathcal{F}$ an incremental multiplier matrix $M$ for $f$,
		and   a positive scalar $\bar{\alpha}$  such that \eqref{eq:barPhiLMI}
		holds and
		\begin{subequations}
			\begin{align}
			\label{eq:lmi_nonlin_mu_indep_f2}
			\mathcal{P}\mathcal{B}+\mathcal{Y}\mathcal{D} &=  \mathcal C^\top \mathcal F^\top \\
			\mathcal{D}_q +L_2\mathcal{D}  &=0										\label{eq:DqD}\\
			\mathcal{C}^\top\mathcal{F}^\top\mathcal{F}\mathcal{D} &=0, \qquad \mathcal{C}^\top \mathcal{F}^\top\mathcal{F}\mathcal{D}_f = 0 	\label{eq:C'D=0}
			\end{align}		
		\end{subequations}
		Consider any matrix  $H\in \mathbb{R}^{n_z\times n_x}$  and any performance level $\gamma>0$. Considering any positive $\alpha < \bar{\alpha}$, choose $\mu_1 >0$ to satisfy
		\begin{equation}
		\label{eq:mu2}
		\mu_1 H^\top H  \preceq  2(\bar{\alpha} -\alpha)\mathcal{P} 
		\end{equation}
		and choose $\zeta$ to satisfy
		\begin {equation}
		\label{eq:zeta}
		\zeta \ge   1/(2\mu_1 \gamma^2)
	\end{equation}
	Consider now  observer~\eqref{eq:obs} with gains  $L_2$ and
	\begin{equation}
	\label{eq:L1small}
	L_1 =  \mathcal P^{-1} \left(\mathcal Y - \zeta \mathcal C^\top \mathcal{F}^\top \mathcal F\right)
	\end{equation}
	Then, for any initial condition $e(t_0) = e_0$  with $t_0\in \mathbb{R}$ and $e_0 \in \mathbb{R}^{n_e}$ and any  ${\mathcal L}_2$ input $v(\cdot):[t_0, \infty) \rightarrow \mathbb{R}^{n_v}$,
	inequalities
	\eqref{eq:L_2}
	and
	\eqref{eq:eBound}
	hold
	for all $t\ge t_0$.
	Hence the  error dynamics with   performance output $z = He$ are \ltwo-stable with performance level 
	$\gamma$.
\end{theorem}

\begin{proof}
	Consider any matrix $H\in \mathbb{R}^{n_z\times n_x}$ and any scalar $\gamma >0$. 
	Letting  $\mu_2 =\gamma^2 \mu_1$, we have $\zeta  \ge   1/2\mu_2$ and 
	we now show that \eqref{eq:thm1a} holds with $\mathcal{Y}$  replaced with 
	\begin{equation}
	\tilde{\mathcal Y} = \mathcal Y  -\zeta  \mathcal C^\top \mathcal{F}^\top\mathcal{F}
	\end{equation} 
	We saw from the proof of Lemma \ref{thm:finite_gamma_nonlin} that~\eqref{eq:thm1a}  (with $\tilde{\mathcal Y}$ replacing  $\mathcal{Y}$) is equivalent to  
	\begin{equation}
	\label{eq:tildePhiLMI}
	\begin{bmatrix}
	\tilde{\Phi} 	&	\mathcal P\mathcal B_f	+\tilde{\mathcal{Y}}\mathcal{D}_f\\
	\star		&0
	\end{bmatrix}
	+\bar \Gamma^\top  M \bar \Gamma \preceq 0,
	\end{equation}
	where
	\begin{align}
	\tilde{\Phi} &=\mathcal P \mathcal A + \mathcal A^\top \mathcal P + \tilde{\mathcal Y} \mathcal C + \mathcal C^\top \tilde{\mathcal Y}^\top + 2\alpha \mathcal P 
	+\mu_1H^\top H +  \Xi_{13}N^{-1}\Xi_{13}^\top  \nonumber   \\
	&\le \mathcal P \mathcal A + \mathcal A^\top \mathcal P + {\mathcal Y} \mathcal C + \mathcal C^\top\mathcal Y^\top -2\zeta\mathcal{C}^\top\mathcal{F}^\top\mathcal{F}\mathcal{C}
	+ 2\bar{\alpha} \mathcal P 
	+  \Xi_{13}N^{-1}\Xi_{13}^\top					\nonumber\\
	&= \Phi_{\bar \alpha} -2\zeta\mathcal{C}^\top\mathcal{F}^\top\mathcal{F}\mathcal{C}
	+  \Xi_{13}N^{-1}\Xi_{13}^\top 	
	\end{align}
	with
	$\Phi_{\bar \alpha}$   given by \eqref{eq:PhiAlpha}
	and 
	\begin{align*}
	\Xi_{13} &= \mathcal{P}\mathcal{B}+\tilde{\mathcal Y}\mathcal{D}  \\
	&= \mathcal{P}\mathcal{B}+\mathcal{Y}\mathcal{D} -\xi\mathcal{C}^\top\mathcal{F}^\top\mathcal{F}\mathcal{D}\\ &=\mathcal{C}^\top\mathcal{F}
	\end{align*}
	The last two equalities follow from \eqref{eq:C'D=0} and \eqref{eq:lmi_nonlin_mu_indep_f2}.
	Also, using \eqref{eq:N} and \eqref{eq:DqD}, $N=\mu_2 I$.
	Note that 
	\begin{equation}
	\label{eq:YDf}
	\tilde{ \mathcal Y}\mathcal{D}_f = { \mathcal Y}\mathcal{D}_f  -\zeta \mathcal{C}^\top  \mathcal{F}^\top \mathcal{F}\mathcal{D}_f = { \mathcal Y}\mathcal{D}_f
	\end{equation}
	We now obtain that
	\begin{align*}
	\Xi_{13}N^{-1}\Xi_{13}^\top &= \mu_2^{-1} \mathcal{C}^\top \mathcal F^\top \mathcal{F}\mathcal{C}\\
	&\le 
	2\zeta \mathcal{C}^\top \mathcal F^\top \mathcal{F}\mathcal{C}
	\end{align*}
	and
	$\tilde{\Phi} \preceq \Phi_{\bar \alpha}$.
	It now follows from \eqref{eq:barPhiLMI} and \eqref{eq:YDf} that \eqref{eq:thm1a} holds.
	The proof is completed by invoking Theorem~\ref{thm:1}.
\end{proof}

\begin{remark}
	Theorem~\ref{thm:uio_reject} implies that, for any  $H$,
	$H\xi$
	can be estimated to arbitrary accuracy. 
	That is, for any given $\varepsilon>0$ there exists a corresponding observer of the form~\eqref{eq:obs} that is \ltwo-stable with performance level 
	$\varepsilon/\|v(\cdot)\|_2$. 
	From Definition~\ref{defn1}, we deduce that, for zero initial state,
	$
	\|H e(\cdot)\|_2 \le \varepsilon$.
\end{remark}

\section{Linear error dynamics}
Consider plant \eqref{eq:plant2} with $f_1=0$ and disturbance model \eqref{eq:internal_model_w} with $f_2 =0$, that is,
\begin{subequations}
	\label{eq:plant2lin}
	\begin{align}
	\dot x &= A x + g_x(t,y) +B w, \;\; y= C x+  g_y(t,y) + D w,
	\end{align}
\end{subequations}
and
\begin{subequations}
	\label{eq:internal_model_wlin}
	\begin{align}
	\dot x_m &= A_m x_m  + B_m v,\;\; w = C_m x_m  + D_m v.
	\end{align}
\end{subequations}

The  corresponding observer~\eqref{eq:obs} simplifies to
\begin{subequations}
	\label{eq:obs_lin}
	\begin{align}
	\dot{\hat \xi} &= \mathcal A\hat{\xi} + \tilde{g}_\xi(t, y) + L_1(\hat y - y)\\
	\hat{y} & = \mathcal{C}\xi + \tilde{g}_y(t,y).
	\end{align}
\end{subequations}
which only involves the observer gain matrix $L_1$. The error dynamics resulting from this  observer are described by
\begin{equation}
\label{eq:err_dyn_lin}
\dot e = (\mathcal A + L_1\mathcal C) e - (\mathcal B+L_1 \mathcal{D}) v.
\end{equation}

Herein, we obtain simple conditions guaranteeing the existence of  an observer gain $L_1$  which yields the desired behavior.
First we need a preliminary lemma.

\begin{lemma}
	\label{lem:detect}
	A pair $(\mathcal{A}, \mathcal{C})$ is detectable if and only if 
	there are matrices $\mathcal P = \mathcal P^\top  \succ 0$, $\mathcal{Y}$ and a scalar $\bar{\alpha} >0$ such that
	\begin{equation}
	\label{eq:obsLinLMI}
	\Phi_{\bar \alpha} \prec 0
	\end{equation}
	where $\Phi_{\bar \alpha}$ is given by \eqref{eq:PhiAlpha}.
\end{lemma}

\begin{proof}
	Detectability of  $(\mathcal{A}, \mathcal{C})$   is equivalent to the existence of a matrix $L_1$ such
	that  $\mathcal{A} + L_1\mathcal{C}$ is Hurwitz, that is, all of its eigenvalues have negative real part.
	By Lyapunov theory this is   equivalent to the  existence of 
	a matrix ${\mathcal P}={\mathcal P}^\top  > 0$ such that
	\[
	\mathcal P (\mathcal{A} + L_1\mathcal{C}) + (\mathcal{A} + L_1\mathcal{C} )^\top  \mathcal P  +2 I = 0.
	\]
	Choosing $\bar\alpha >0$ such that $\bar\alpha \mathcal P \prec I$ results in
	\begin{equation}
	\label{eq:Lyap1}
	\mathcal P (\mathcal{A} + L_1\mathcal{C}+ \bar\alpha I ) + (\mathcal{A} + L_1\mathcal{C}+ \bar\alpha I )^\top  \mathcal P   \prec 0,
	\end{equation}
	that is,
	\eqref{eq:obsLinLMI}
	with $\mathcal{Y} = \mathcal{P}L_1$.
	Conversely, if \eqref{eq:Lyap1} holds, then, by Lyapunov theory, $\mathcal{A} + L_1\mathcal{C}$ is Hurwitz .
\end{proof}

\begin{lemma}
	\label{lem:detect1}
	Suppose that  $(A, C)$ and $(A_m, C_m)$ are detectable
	and Condition \ref{cond1} holds.
	Then   $(\mathcal{A}, \mathcal{C})$  is detectable.
\end{lemma}

\begin{proof}
	The pair   $(\mathcal{A}, \mathcal{C})$ is detectable if  and only if
	\begin{equation*}
	\label{eq:rank_pf_model} 
	{H}(\lambda) =
	\begin{bmatrix}
	\mathcal{A}-\lambda I \\ \mathcal{C} 
	\end{bmatrix}
	=
	\begin{bmatrix}
	A - \lambda I & B C_m \\
	0 & A_m -\lambda I\\
	C & DC_m
	\end{bmatrix}
	\end{equation*}
	has full column rank
	for every  eigenvalue  $\lambda$ of $\mathcal{A}$ with non-negative real part.
	Note  that $\lambda$ is an eigenvalue of $\mathcal{A} $ if and only if it is  an eigenvalue of  $A$ or $A_m$.
	
	Suppose that $H(\lambda)$ does not have full column rank. Then there is a non-zero vector 
	$\xi = \begin{bmatrix}
	x^\top & x_m^\top 
	\end{bmatrix}^\top $ such that $H(\lambda)\xi = 0$, that is,
	\begin{align}
	(A- \lambda I)x + BC_m x_m&=0	\label{eqH=01}		\\
	Cx + DC_m x_m	&=0				\label{eqH=02} \\ 
	(A_m - \lambda I) x_m &=0.
	\label{eqH=03}
	\end{align}
	
	If $x_m=0$ then $x\neq 0$ and the above equations imply that
	$(A- \lambda I)x =0$ and $Cx=0$.
	Thus $\lambda$ is an eigenvalue of $A$  and
	the matrix 
	$\left[
	\begin{smallmatrix}
	A - \lambda I \\ C
	\end{smallmatrix} 
	\right]$
	does not have full column rank.
	Since $(C,A)$ is detectable,  the real part of $\lambda$ must be negative.

	If $x_m \neq 0$, equation~\eqref{eqH=03} implies that $\lambda$ is an eigenvalue of $A_m$.
	If $C_mx_m =0$ then	
	the matrix 
	$\left[
	\begin{smallmatrix}
	A _m- \lambda I \\ C_m
	\end{smallmatrix} 
	\right]$
	does not have full column rank.
	Since $(A_m, C_m)$ is detectable,  the real part of $\lambda$ must be negative.
	If $C_mx_m\neq 0$, equations~\eqref{eqH=01} and~\eqref{eqH=02} imply that
	\[
	\begin{bmatrix}
	A-\lambda I & B \\ C & D
	\end{bmatrix} \begin{bmatrix}
	x \\ C_mx_m
	\end{bmatrix} = 0.
	\]
	that is the matrix
	$
	\left[\begin{smallmatrix}
	A-\lambda I & B \\ C & D
	\end{smallmatrix} \right]
	$
	does have full column rank.
	Since $\lambda$ is an eigenvalue of $A_m$, $\lambda$ must have negative real part.
	
	Thus, we have shown that if  $H(\lambda)$ does not have full column rank then the real part of $\lambda$ is negative.
	Hence $H(\lambda)$ has full column rank whenever the real part of $\lambda$ is non-negative and $(\mathcal A, \mathcal C)$ is detectable.
\end{proof}

\begin{theorem}\label{thm:ACdetect_finite_gamma}
	Suppose that  $(A, C)$ and $(A_m, C_m)$ are detectable 
	and Condition \ref{cond1} holds
	Then, for any performance output $z=He$ there exists 
	an observer gain $L_1$ such that the observer error dynamics~\eqref{eq:err_dyn_lin} are \ltwo-stable with some performance level $\gamma$.
\end{theorem}
\begin{proof}
	Note that~\eqref{eq:barPhiLMI} of Lemma~\ref{thm:finite_gamma_nonlin} with $\mathcal B_f =0$ and $M=0$ is equivalent to
	\eqref{eq:obsLinLMI}.
	Hence, using Theorem \ref{thm:1}, Lemma~\ref{thm:finite_gamma_nonlin} and Lemma \ref{lem:detect} we only need to show that
	$(\mathcal A, \mathcal C)$ is detectable. This follows from Lemma \ref{lem:detect1}.
\end{proof}


\subsection{Estimating  with arbitrarily small error}

First we have the following result from \cite{corless98}.

\begin{lemma}
	\label{lem:trans2}
	Suppose $ \mathcal{A}\in \mathbb{R}^{n_\xi \times n_\xi}$, $ \mathcal{B}\in \mathbb{R}^{n_\xi \times n_w}$ and  $ \mathcal{C}\in \mathbb{R}^{n_y \times n_\xi}$.
	Then there exist matrices $\mathcal P= \mathcal P^\top \succ 0 $, $ \mathcal Y$ and    $\mathcal F$ 
	such that 
	\begin{align}
	\label{eq:obsLinLMI2at} 
	\mathcal P \mathcal A + \mathcal A^\top \mathcal P +\mathcal Y \mathcal C + \mathcal C^\top \mathcal Y^\top &\prec 0	\\
	\label{eq:obsLinLMI2bt} 
	\mathcal B^\top \mathcal P &= \mathcal F\mathcal C
	\end{align}
	if and only if
	\begin{equation}
	\label{eq:rankCBt}
	\mathrm{rank}\, \mathcal C\mathcal{B} = \mathrm{rank}\,  \mathcal{B}
	\end{equation}
	and
	\begin{equation}
	\label{eq:rankSystemt}
	\mathrm{rank}\,
	\begin{bmatrix}
	\mathcal{A} - \lambda I	&  \mathcal{B}\\
	\mathcal{C}	&  0
	\end{bmatrix}
	=n_\xi+  \textrm{rank}\, \mathcal{B}
	\end{equation}
	for all $\lambda \in \mathbb{C}$ with non-negative real part.	
\end{lemma}

The following result provides conditions that, when satisfied, ensure the existence of observers of the form~\eqref{eq:obs_lin} that generates error dynamics that are~\linf-stable with any arbitrary performance level $\gamma>0$.
\begin{lemma}\label{lemma8}
	Suppose $DD_m=0$
	and
	\[CBD_m + DC_mB_m
	,\;\; \begin{bmatrix}
	A -\lambda I	&B		\\
	C				&D
	\end{bmatrix} , \;\;
	\begin{bmatrix}	A_m -\lambda I	&B_m		\\
	C_m				&D_m
	\end{bmatrix}
	\]
	have full column rank for all $\lambda \in \mathbb{C}$ with non-negative real part.
	Consider any matrix  $H\in \mathbb{R}^{n_z\times n_x}$  and any performance level $\gamma >0$.
	Then there exist matrices $\mathcal P= \mathcal P^\top  \succ 0 $, $ \mathcal Y$ and    $\mathcal F$  
	such that ~\eqref{eq:obsLinLMI2at} and ~\eqref{eq:obsLinLMI2bt} hold.
	Choose $\mu_2 >0$ and $\zeta$ to satisfy~\eqref{eq:mu2}
	and
	\eqref{eq:zeta}.
	Then the observer~\eqref{eq:obs_lin} with gain 
	given by~\eqref{eq:L1small}
	generates error dynamics with   performance output $z = He$ that are\ltwo-stable with performance level $\gamma$.
\end{lemma}

\begin{proof}
	We use Theorem \ref{thm:uio_reject} and Lemma \ref{lem:trans2}.
	Since $\mathcal{B}_f = \mathcal{D}_f =0$ and considering $M=0$,   
	\eqref{eq:barPhiLMI} 
	reduces to
	\eqref{eq:obsLinLMI}.
	The existence of $\bar{\alpha} >0$ such that \eqref{eq:obsLinLMI} holds is equivalent to 	\eqref{eq:obsLinLMI2at} of Lemma \ref{lem:trans2}.
	Also $\mathcal{D}_q = 0$ and $\mathcal{D} = DD_m =0$; this implies that 
	\eqref{eq:lmi_nonlin_mu_indep_f2}-\eqref{eq:C'D=0}
	reduce to \eqref{eq:obsLinLMI2bt}.
	Since 
	\[
	\mathcal{C}{\mathcal B} = CBD_m + DC_mB_m,
	\]
	has full column rank, condition \eqref{eq:rankCBt} holds.
	Also, $\mathcal{B}$ must have full column rank, that is, $n_v$.
	To verify condition \eqref{eq:rankSystemt} of Lemma~\ref{lem:trans2}, consider any $\lambda \in \mathcal{C}$ with non-negative real part.
	Then
	\begin{align*}
	\begin{bmatrix}
	\mathcal{A} - \lambda I	&\mathcal{B}		\\
	\mathcal{C}	&  0 
	\end{bmatrix}
	&=
	\begin{bmatrix}
	A - \lambda I	&0	&  B\\
	0	&I	&0			\\
	C	&  0				&D
	\end{bmatrix}
	\begin{bmatrix}
	I	&0	&  0\\
	0	&A_m -\lambda I	&B_m			\\
	0	&  C_m				&D_m
	\end{bmatrix}.
	\end{align*}
	As a consequence of the hypotheses of the lemma, the two matrices on the right-hand side of the second equality have maximum column rank;
	hence $		\left[\begin{smallmatrix}
	\mathcal{A} - \lambda I	&\mathcal{B}		\\
	\mathcal{C}	&  0 
	\end{smallmatrix}\right]$ has maximum column rank, that is, $n_\xi +n_v$  which equals $n_\xi + \mbox{rank}\;\mathcal B$.
	condition~\eqref{eq:rankSystemt}. Invoking Theorem~\ref{thm:uio_reject} and Lemma~\ref{lem:trans2} concludes the proof.
\end{proof}

\subsubsection{Connection to classical rank conditions}
Consider the classical linear  case of the linear system, $\dot x = Ax + Bw$, $y=Cx$, and $w=v$. This is described by~\eqref{eq:plant2lin} and~\eqref{eq:internal_model_wlin} with $D=0$, $D_m=I$ and $A_m,B_m,C_m$ vanish. 
Hence,
\begin{align*}
CBD_m + DC_mB_m &= CB\\
\begin{bmatrix}
\mathcal{A} - \lambda I	&\mathcal{B}		\\
\mathcal{C}	&  0 
\end{bmatrix}
&=
\begin{bmatrix}
A -\lambda I	&B		\\
C				&D
\end{bmatrix}, \\
\begin{bmatrix}
A_m -\lambda I	&B_m			\\
C_m				&D_m
\end{bmatrix}&=
-\lambda I\,.
\end{align*}
Consequently,  the conditions in Lemma~\ref{lemma8} reduce to the requirements  that 
$CB$ and 
$
\left[\begin{smallmatrix}
A -\lambda I	&B		\\
C				&0
\end{smallmatrix} \right]
$
have full column rank for all $\lambda \in \mathbb{C}$ with non-negative real part.
With $B$ full column rank, these are exactly  the  classical  conditions  for state estimation to an arbitrary degree of accuracy; see~\cite{corless98}.

\section{Proof of Theorem \ref{thm:1}}
\label{sec:Th1Proof}

First we need the following result.
\begin{lemma}\label{lem:linf}
	Consider a system 
	described by
	\eqref{eq:gensys} with state $e(t)\in \mathbb{R}^{n_e}$, input $v(t)\in\mathbb{R}^{n_v}$ and performance output $z(t) \in\mathbb{R}^{n_z}$.
	Suppose there exists a differentiable function $V:\mathbb{R}^{n_e} \rightarrow \mathbb{R}$  and scalars
	$\alpha, \beta_1, \beta_2,  \mu_1 >0$ and $\mu_2 \ge 0$
	such that
	\begin{equation}\label{eq:pf_V}
	\beta_1  \|e\|^2  \le V(e) \le \beta_2 \|e\|^2
	\end{equation}
	and
	\begin{equation}
	\label{lem2:condns}
	\mathcal D V(e)\, F(t,e,v) \le -2\alpha V(e) - \mu_1\|G(t,e,v)||^2 +\mu_2||v||^2
	\end{equation}
	for all $t\in \mathbb{R}$, $e\in\mathbb{R}^{n_e}$ and $v\in\mathbb{R}^{n_v}$, 
	where $\mathcal D V$ denotes the derivative of $V$.
	Then, for any initial condition $e(t_0) = e_0$  with $t_0\in \mathbb{R}$ and $e_0 \in \mathbb{R}^{n_e}$ and any $\mathcal{L}_2$ exogenous input $v(\cdot):[t_0, \infty) \rightarrow \mathbb{R}^{n_v}$, inequalities
	\eqref{eq:L_2}  and
	\eqref{eq:eBound} hold 
	for all $t\ge t_0$.
	Hence, system~\eqref{eq:gensys} is globally uniformly \ltwo-stable with performance level 
	$\gamma=\sqrt{\mu_2/\mu_1}$.
\end{lemma}
\begin{proof}
	Consider any initial condition $e(t_0) = e_0$ and any $\mathcal{L}_2$ exogenous input $v(\cdot):[t_0, \infty) \rightarrow \mathbb{R}^{n_v}$.
	Recalling~\eqref{lem2:condns}, the time-derivative of $V(e)$ evaluated  along a  corresponding trajectory of  \eqref{eq:gensys}   satisfies
	\begin{align}
	\frac{dV(e(t))}{dt} &= \mathcal D V(e(t))\dot{e}(t) =\mathcal D V(e(t))\, F(t,e(t),v(t)) \nonumber\\
	&\le -2\alpha V(e(t)) - \mu_1\|z(t)\|^2 + \mu_2\|v(t)\|^2		\label{eq:Vdot0}\\		
	&\le -2\alpha V(e(t)) + \mu_2\|v(t)\|^2
	\label{eq:Vdot}
	\end{align}
	for all $t\ge t_0$.
	Hence,
	\begin{align*}
	V(e(t))  &\le   e^{-2 \alpha (t-t_0)} V(e_0)
	+  \mu_2 \int_{t_0}^t e^{-2 \alpha (t-\tau)} \|v(\tau)\|^2d \tau
	\\
	&\le e^{-2 \alpha (t-t_0)}\beta_2 \|e_0\|^2 +\mu_2  \|v(\cdot)\|_2^2
	\end{align*}
	Since $\beta_1 \|e\|^2 \le V(e)$ we see that
	\begin{align}
	\|e(t)\|^2   \le e^{- 2\alpha (t-t_0)}(\beta_2/\beta_1) \|e_0\|^2 +(\mu_2/\beta_1)  \|v(\cdot)\|^2_2,
	\end{align}
	from which it follows that
	\begin{align}
	\|e(t)\|   \le e^{- \alpha (t-t_0)}\sqrt{\beta_2/\beta_1} \|e_0\| +\sqrt{\mu_2/\beta_1}  \|v(\cdot)\|_2.
	\end{align}
	
	To demonstrate  \eqref{eq:L_2}, 
	note that \eqref{eq:Vdot0}  implies that
	\begin{align}
	\label{eq:Vdot2}
	\mu_1\|z(t)\|^2 \le  - \frac{dV(e(t))}{dt} + \mu_2\|v(t)\|^2,
	\end{align}
	which, upon integrating from $t_0$ to any $t\ge t_0$    results in
	\begin{align*}
	\mu_1  \int_{t_0}^t  \|z(\tau)\|^2d \tau 
	&\le V(e_0) + \mu_2 \|v(\cdot)\|_2^2.
	\end{align*}
	Hence, for all $t\ge t_0$,
	\begin{align*}
	\int_{t_0}^t  \|z(\tau)\|^2d \tau 
	&\le \mu_1^{-1} V(e_0) +\mu_2/\mu_1\|v(\cdot)\|_2^2		\\
	&\le \beta_2/\mu_1 \|e_0\|^2 + \mu_2/\mu_1\|v(\cdot)\|_2^2.	
	\end{align*}
	This implies that  $z(\cdot)$ is an $\mathcal{L}_2$ signal and 
	\begin{equation}
	\|z(\cdot)\|_2 \le \sqrt{\beta_2/\mu_1} \|e_0\| + \sqrt{\mu_2/\mu_1}\|v(\cdot)\|_2,
	\end{equation}%
	which concludes the proof.
\end{proof}

Consider now an input-output system described by
\begin{align}
\dot e&= \tilde{\mathcal A }e + \tilde{\mathcal B}_f \tilde{f} +\tilde{\mathcal B} v			\label{eq:io2a}, \quad z= He,
\end{align}
and suppose there is a  symmetric matrix $M$ so that the term $\tilde{f}$ satisfies
\begin{equation}
\label{eq:iqct}	\begin{bmatrix}
\tilde{q} \\ \tilde{f}
\end{bmatrix}^\top M \begin{bmatrix}
\tilde{q}\\ \tilde{f}
\end{bmatrix} \ge 0
\quad  \mbox{with}\quad
\tilde{q}= \tilde{\mathcal C}_q e+ \tilde{\mathcal D}_{qf} \tilde{f} + \tilde{\mathcal{D}}_{q}  v 
\end{equation}
for all $t\ge 0$, $e\in\mathbb{R}^{n_e}$ and $v\in\mathbb{R}^{n_v}$.
Then we have the following result.%

\begin{lemma}
	\label{lem:Lemma3}
	Consider system \eqref{eq:io2a}  satisfying \eqref{eq:iqct} and suppose that there is  a  matrix  $\mathcal P=\mathcal P^\top \succ 0$
	and   scalars $\alpha, \mu_1 >0$, $\mu_2 \ge 0$,
	such that 
	\begin{align}
	\label{eq:theo1_a}
	\begin{bmatrix}
	{\mathcal P} \tilde{\mathcal A}  +\tilde{\mathcal A}^\top \mathcal P+  2\alpha { \mathcal P} + \mu_1H^\top H
	& { \mathcal P}\tilde{\mathcal B}_f & {\mathcal P}\tilde{\mathcal B} \\
	\star & 0 & 0\\
	\star & 0 & -\mu_2I
	\end{bmatrix}+ \tilde{\Gamma}^\top M \tilde{\Gamma} & \preceq 0,
 	\end{align}%
	where
	\begin{equation}
	\tilde{\Gamma} = \begin{bmatrix}\tilde{\mathcal C}_q  & \tilde{\mathcal D}_{qf} & \tilde{\mathcal D}_q   \\ 0& I & 0 
	\end{bmatrix}.
	\label{eq:z}
	\end{equation}
	Then, for any initial condition $e(t_0) = e_0$  with $t_0\in \mathbb{R}$ and $e_0 \in \mathbb{R}^{n_e}$ and any $\mathcal{L}_2$  input $v(\cdot):[t_0, \infty) \rightarrow \mathbb{R}^{n_v}$, inequalities
	\eqref{eq:L_2}  and
	\eqref{eq:eBound} hold 
	for all $t\ge t_0$ with $\beta_1 = \lambda_{\min}( { \mathcal P})$ and $\beta_2 = \lambda_{\max}( { \mathcal P})$.
	Hence, system~\eqref{eq:io2a} 
	is \ltwo-stable with  performance level $\gamma=\sqrt{\mu_2/\mu_1}$.
\end{lemma}
\begin{proof}
	We will show that  system \eqref{eq:io2a}-\eqref{eq:iqct} 
	satisfies the hypotheses of Lemma~\ref{lem:linf} with $V(e) = e^\top  { \mathcal P}  e$ . 
	This choice of $V$ satisfies the Rayleigh inequality
	$$\lambda_{\min}( { \mathcal P})\|e\|^2 \le V(e)\le \lambda_{\max}( { \mathcal P})\|e\|^2$$
	for all $e\in\mathbb{R}^{n_e}$. Hence,~\eqref{eq:pf_V} holds with $\beta_1 = \lambda_{\min}( { \mathcal P}) >0$ and $\beta_2 = \lambda_{\max}( { \mathcal P})$.
	For system~\eqref{eq:io2a}-\eqref{eq:iqct}, $$F(t,e,v) =  \tilde{\mathcal A }e +\tilde{\mathcal B}_f \tilde{f} +\tilde{\mathcal B} v.$$
	Therefore,
		\begin{align}\label{pf1:1}
		\mathcal D V(e)F(t,e,v) &= 2e^\top { \mathcal P}(\tilde{\mathcal A }e +\tilde{\mathcal B}_f \tilde{f} +\tilde{\mathcal B} v),		\nonumber\\
		&= e^\top( { \mathcal P}  \tilde{\mathcal A } +  \tilde{\mathcal A }^\top { \mathcal P}) e + 2 e^\top\tilde{\mathcal B}_f \tilde{f} + 2e^\top\tilde{\mathcal B} v.
		\end{align}%
	Recalling the description of $\tilde{q}$ in~\eqref{eq:iqct}, we see that
	$\tilde{\Gamma} \begin{bmatrix}
	e^\top &  \tilde{f}^\top & v^\top
	\end{bmatrix}^\top = \begin{bmatrix}
	\tilde{q}^\top & \tilde{f}^\top
	\end{bmatrix}^\top$.
	Hence, pre-and post-multiplying the matrix inequality~\eqref{eq:theo1_a} by
	$ \begin{bmatrix}
	e^\top  & \tilde{f}^\top & v^\top
	\end{bmatrix} 
	$
	and its transpose results in
	{\small
		\begin{align}
		\label{eq:Vdot}
		\nonumber &\mathcal D V(e)F(t,e,v) + 2\alpha V  +\mu_1\|z\|^2 - \mu_2 \|v\|^2 + 
		\begin{bmatrix}
		\tilde{q} \\ \tilde{f}
		\end{bmatrix}^\top M \begin{bmatrix}
		\tilde{q} \\ \tilde{f}
		\end{bmatrix}
		\le 0.
		\end{align}}%
	It now follows from \eqref{eq:iqct} that
	\[
	\mathcal D V(e)F(t,e,v) \le - 2\alpha V (e) + \mu_1\|z\|^2 -\mu_2 \|v\|^2,
	\]
	that is,~\eqref{lem2:condns} holds.
	Using Lemma~\ref{lem:linf}, we are done.
\end{proof}

\subsection{Proof of Theorem \ref{thm:1}}
With the estimation error given by
\eqref{eq:error},   it follows from \eqref{eq:obs} and \eqref{eq:aug_sys}  that the observer error dynamics are given by
\begin{subequations}
	\label{eq:err_dyn}
	\begin{align}
	\dot e &= (\mathcal A + L_1 \mathcal C)e +( \mathcal B_f+L_1\mathcal{D}_f) \tilde f -(\mathcal B +L_1\mathcal{D})v,\\
	\tilde f &= f(t,y,q+\tilde q) - f(t,y,q),\\
	\tilde q &= (\mathcal C_q + L_2 \mathcal C) e + \mathcal{D}_{qf} \tilde f -(\mathcal{D}_q+L_2 \mathcal{D})v.
	\end{align}
\end{subequations}
That is, it is described by \eqref{eq:io2a}  and satisfies \eqref{eq:iqct} with 
{\small
	\begin{align*}
	\tilde{A} &= \mathcal{A}+L_1 \mathcal{C}, \quad \tilde{B}_f =  \mathcal B_f+L_1\mathcal{D}_f,\quad 
	\tilde{B} = -(\mathcal B +L_1\mathcal{D}), \\
	\tilde{\mathcal C}_q &= \mathcal{C}_q + L_2\mathcal{C},\quad 
	\tilde{\mathcal D}_{qf} = \mathcal{D}_{qf} +L_2\mathcal{D}_f, \quad  \mathcal{D}_q= -(\mathcal{D}_q+L_2 \mathcal{D}).
	\end{align*}}%
Recalling that $L_1 = \mathcal P^{-1}\mathcal Y,$ we see that~\eqref{eq:thm1a} is the same as~\eqref{eq:theo1_a}. The desired result now follows from Lemma~\ref{lem:Lemma3}.


\section{Numerical Example}
We employ a modified model of the active magnetic bearing system investigated in~\cite{tsiotras2005low}. 
The modification includes disturbance  inputs and measurement  noise  to illustrate the unknown input observer capabilities and to make the problem more challenging than the one considered  in our previous work~\cite{Chakrabarty2016}. The model is given by
\begin{equation}
\label{eq:ex3_model}
\dot{x} = \begin{bmatrix}
x_2 + w_1\\
x_3+x_3|x_3| 
\\ w_2 
\end{bmatrix},
\qquad y = \begin{bmatrix}
x_1 + 0.1w_1 \\ x_2
\end{bmatrix},
\end{equation}
which is in the form of~\eqref{eq:plant2} with $f_1(t,y,q_1) =q_1|q_1|$ and $q_1 = x_3$.
Considering $$w_1(t) = 1/\sqrt{1+t} \qquad \text{  and  } \qquad  w_2(t) = \log(1+t),$$
$w$ is unbounded in the $\mathcal L_2$ sense, but $\dot w$ is bounded.
Also $w_2$ is unbounded in the usual sense.
Hence $w$ can be modelled by \eqref{eq:simpleW} where $v = \dot{w}$.
Any matrix of the form
$$ M = \kappa\begin{bmatrix}
0 & 1\\ 1 & 0
\end{bmatrix}$$
with any $\kappa\ge 0$
is incremental multiplier matrix for $f_1$.
Note that we will solve for $\kappa$: we only know the form of $M$, the parameter $\kappa$ is an optimization variable. We choose $z=w$, which implies 
that we are interested in obtaining a good estimate of $w$ and are ready to accept lower accuracy when reconstructing $x$. Thus, $z=w$.
We  fix $L_2 = \begin{bmatrix}
0 & -110
\end{bmatrix}$ and solve~\eqref{eq:gevp} with a line search to find an optimal $\alpha$. We get
$\alpha^\star = 0.710$, $\kappa=1.6\times 10^6$, and $\mu_2^\star = 0.08$. 
We test our proposed observer on system~\eqref{eq:ex3_model} with the initial conditions $$x(0)= 
\begin{bmatrix}
-2.7247 &  10.9842 &  -2.7787
\end{bmatrix}^\top$$ and $\hat{\xi}(0) =  0.$

\begin{figure}[!h]
	\centering	\includegraphics[width=.55\columnwidth]{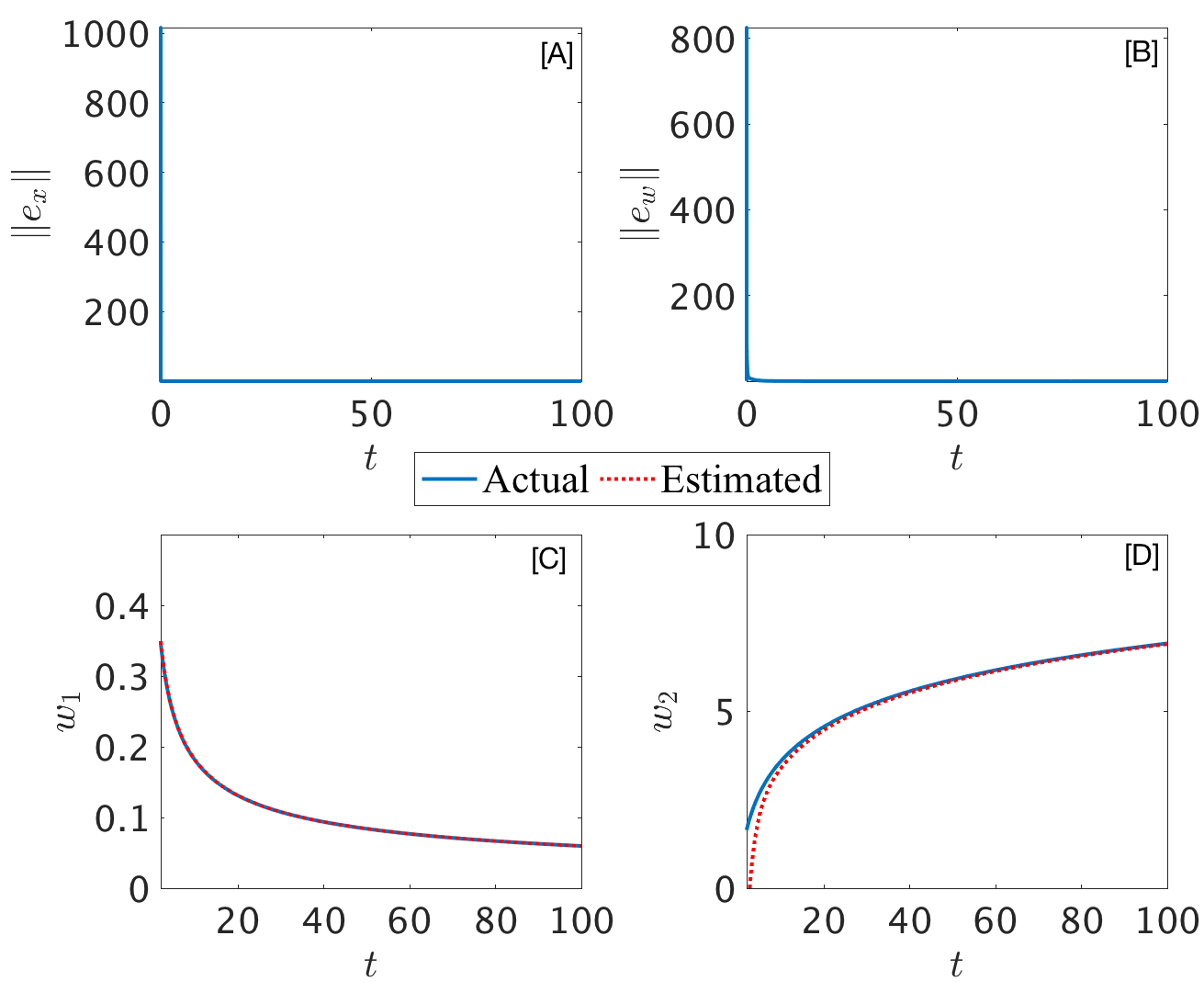}
	\caption{[A] State estimation error $e_x\triangleq \hat x - x$ of the nonlinear system in~\eqref{eq:ex3_model}. [B] Unknown input estimation error. The convergence of the norm of $e_w \triangleq \hat w- w$ is illustrated. [C, D] Unknown inputs (blue) and their estimates (dashed red).}
	\label{fig:ex3}
\end{figure}
The response of the proposed observer is shown in Figure~\ref{fig:ex3}. Note that the unknown input $w_2$ is monotonically increasing, yet from Figure~\ref{fig:ex3}[C-D], we observe that the estimates of the unbounded unknown inputs are very accurate; this is to be expected since $\gamma$ is small. 

\section{Conclusions}\label{sec:conc}
This paper provides an LMI based approach to the design of observers for estimating the state and unknown exogenous input for a wide range of nonlinear systems.
The resulting input-output system from exogenous input to estimation error is $\mathcal{L}_2$ stable with a gain that can be pre-specified and computed via standard toolboxes.

\bibliographystyle{IEEEtran}
\bibliography{../refs}

\end{document}